\documentclass[[submission,copyright,creativecommons]{eptcs}
\usepackage[utf8]{inputenc}
\usepackage{amsthm}
\usepackage[T1]{fontenc}
\usepackage{underscore}
\usepackage{graphicx}
\usepackage{amsfonts}
\usepackage{bussproofs}
\usepackage[fleqn]{amsmath}
\usepackage{bookmark}
\usepackage{lineno}
\usepackage{caption}

\usepackage{listings}
\newcommand{\code}[1]{\lstinline[mathescape]|#1|}

\lstnewenvironment{curry}{\lstset{
        language=Haskell,
        morekeywords={free},
        deletekeywords={True,False,last,id,and,not},
        escapeinside={(*}{*)}
}}{}
\lstnewenvironment{ccode}{\lstset{
        language=C
}}{}

\newtheorem{theorem}{Theorem}

\usepackage{tikz-cd}

\newcommand{\atom}    [1]{\AxiomC{#1}}
\newcommand{\unary}   [1]{\UnaryInfC{#1} }
\newcommand{\binary}  [1]{\BinaryInfC{#1} }

\newcommand{\FWD}        {\text{*}}
\newcommand{\To}         {\Rightarrow}
\renewcommand{\vec}   [1]{\overline{#1}}
\newcommand{\tfree}   {\textbf{free}}

\newcommand{\vx}{\ \vec{\text{x}}\ }

\newcommand{\ccase}    [2]{\code{case} $#1$ \code{of} $\ #2$}
\newcommand{\clet}     [2]{\code{let} $#1$ \code{in} $#2$}
\newcommand{\cfree}    [2]{\code{let} $#1$ \code{free in} $#2$}
\def\rice{RICE}
\def\pakcs{Pakcs}
\def\kics{KiCS2}
\def\mcc{MCC}
\def\brassel{Bra{\ss}el}

\def\titlerunning{An Execution Model for RICE}
\def\authorrunning{S. Libby}
\begin{document}
\title{An Execution Model for RICE}
\author{Steven Libby
\institute{University of Portland, 
           Portland OR 97203, USA\\ 
           \email{libbys@up.edu}}
           }

\maketitle

\begin{abstract}
In this paper, we build on the previous work of the \rice~\cite{rice}
compiler by giving its execution model.
We show the restrictions to the FlatCurry language
that were made to produce executable code,
and  present the execution model using operational semantics
similar to Launchbury~\cite{lazySemantics}.
Finally, we show that the execution model conforms with 
the standard operational semantics for Curry~\cite{currySemantics}.
\end{abstract}

\section{Introduction}
Recently there has been a renewed interest in the efficient execution
of functional logic programs~\cite{sprite,verse,currygo,rice}.
This has proven to be a rich area of new ideas.
We look in particular at the \rice{} Curry compiler, which has
shown to produce efficient code~\cite{rice}.
This paper provides an execution model for \rice{}.

Previous work on this compiler showed the execution
of programs using a translation to C code.
While this gets the point across,
it is difficult to reason about the correctness of the implementation.
Instead we take an approach similar to \brassel~\cite{Kics2Theory}.
We begin by showing the execution model for the \rice{} compiler,
and show how it is consistent with 
the standard operational semantics for Curry~\cite{currySemantics}.
The primary contribution of this paper is the execution model.

Our execution model differs form previous works in a few respects.
First we differ from Albert et al. \cite{currySemantics} by
encoding the search strategy into the semantics itself.
While Albert et al \cite{detCurrySemantics} parameterize their semantics
on the search strategy, the implementation is left abstract.
\brassel~\cite{Kics2Theory} gives a fully realized execution model
for the Kics2 compiler, which implements non-determinism
with pull tabbing \cite{Kics2Theory}.
Our work is specifically about the RICE compiler, which implements
non-determinism using backtracking \cite{rice}.
This is done primarily because Kics2's performance degrades
on functions that are not right-linear \cite{rice}.
Our aim is to give a clear understanding of the execution of programs
compiled with RICE.

The rest of the paper is organized as follows.
First, we discuss previous work in the execution of Curry programs
and why we chose a different approach.
Second, we examine the syntax of an intermediate representation of Curry,
and restrict it to a form that provides more efficient execution.
Third, we provide the semantics for our restricted code.
Fourth, we show a correspondence with the original semantics.
Finally, we discuss future work and conclude.

\section{Background}

Curry is a functional logic language;
it has a syntax similar to Haskell,
but extends the language in several ways.
The two important differences for this work
are the addition of non-determinism and free variables.

We can construct a non-deterministic expression 
using the choice operator (\texttt{?}).
For example, in the following code, 
\texttt{pickOne} non-deterministically picks an element in a list, 
while \texttt{member} constrains the choice to determine 
membership in that list.
If the element \texttt x is not in \texttt l, then the member function
simply fails and does not returen a value, so there is no need for an
\texttt{otherwise} case like in Haskell.
\begin{curry}
pickOne (x:xs) = x ? pickOne xs
member x l 
  | x == pickOne l = True
\end{curry}

We can also create free (or logic) variables using the syntax
\texttt{while x free}.
For example, we can create a similar \texttt{member} function 
by constraining a list with free variables to match
our input list.
\begin{curry}
member x l
 | l == first++[x]++last = True
  where first, last free
\end{curry}

Non-determinism can be implemented by picking 
a search strategy~\cite{kics2,mcc},
and free variables can be constrained through
narrowing~\cite{needed,TheoryToCurry}.

Curry originally grew out of the field of term rewriting and narrowing.
Antoy~\cite{DefinitionalTrees} showed that term rewrite systems belonging to 
the class of Inductively Sequential systems could be given an
optimal rewriting strategy,
which was later extended by Antoy et al.~\cite{needed}
to an optimal narrowing strategy.
Later Antoy~\cite{lois} showed how this class of rewriting systems
could be extended to include overlapping rules,
and these Limited Overlapping Inductively Sequential (LOIS) systems
could be computed with a combination of narrowing and a search strategy.
This narrowing plus search proved to be a solid
basis for the Curry language and implementations~\cite{pakcs}.

While narrowing proves to be a strong theoretical foundation for Curry,
there are implementation issues that need to be addressed.
In particular, Curry is a lazy language that contains higher order functions.
While there are theories of higher order rewriting~\cite{HOR},
Curry implementations often tend to use defunctionalization 
to deal with higher order rewriting~\cite{kics2,TheoryToCurry},
thus turning it back into a first order rewriting system.
While this addresses the theoretical issues of higher order functions
nicely, it is not efficient~\cite{verse,evalApply}.
It should be noted that the implementations of \kics\ and \mcc\ do
not rely on defunctionalization~\cite{kics2,mcc}.

Lazy evaluation can also pose an issue with using term rewriting
as a basis for Curry.
Terms have a tree-like structure, and do not support sharing.
If a variable is duplicated in a rewrite rule, the entire
subterm is duplicated.
On the other hand, lazy programs do not ever duplicate expressions.
In lazy functional languages, this is important for efficiency reasons.
However in functional logic languages, this changes the semantics.
Therefore sharing must be preserved at all costs~\cite{sharingProblems}.
For this reason, Curry is often presented 
as a graph rewriting system~\cite{graphRewriting}.

Hanus et al.~\cite{currySemantics} described
an intermediate language called FlatCurry
in order to give a semantics for Curry.
They then described a natural semantics for FlatCurry in the style
of Launchbury~\cite{lazySemantics}.
However they did not specify how choice should be handled,
instead leaving it up to the implementation.
This initial Curry semantics was extended 
to be completely deterministic~\cite{detCurrySemantics}
by parameterising with a search strategy.
\brassel~\cite{Kics2Theory} uses this semantics as a starting point.
He then introduces several transformations to the FlatCurry program
to put it in a form that can be readily translated to Haskell.
Our work is similar to \brassel's work.
We recall the syntax for FlatCurry, 
describe the changes we made, and give the semantics
for our new version of FlatCurry.
Our work differs from previous approaches because we encode the search
strategy into the evaluation itself.
This makes for a more complicated semantics,
but it allows us to examine the efficiency of Curry programs as a whole.
This can lead to new optimizations 
like fast backtracking~\cite{ricePaper,rice}.

\section{FlatCurry}

We begin with a discussion of the 
abstract syntax of FlatCurry~\cite{currySemantics}.
The language itself is similar to Core Haskell
with a few important alterations to support 
the features of functional logic programming.
The syntax is given in Figure \ref{fig:flat}.
The semantics for FlatCurry are recalled in Figure \ref{fig:nat}.
In this semantics, free variables are represented as a variable
that is mapped to itself in the heap $\Gamma[x \mapsto x]$.

The most apparent difference is the addition of the choice operator $?$,
the free variable declaration \cfree{\vec{x}}{e}, and the $\bot$ constant.
The choice operator and free declaration correspond to their counterparts
in the Curry language.  The $\bot$ constant represents a failed computation.
As we will see, $\bot$ is propagated up through the computation until it is
disregarded, or it reaches the root, 
at which point we say the computation fails.

Another peculiarity of this syntax is that there is no general form for
application.
In fact, that are not any lambda expressions in the language.
However, this last part is only an apparent difference
as the compiler has already completed lambda lifting~\cite{lambdalifting} 
by the time it produces FlatCurry code.

The absence of a general application form is slightly more complicated.
We can apply a function to arguments, 
but in this syntax there is no mechanism to apply a function to another
function.
This restriction would make it impossible to support higher order functions.

We solve this problem with a general \texttt{apply} function.
In the expression $(\code{apply}\ e_1\ e_2)$ will evaluate $e_1$ to a function
or partial application and apply it to $e_2$.
In the theory, this is handled with defunctionalization~\cite{TheoryToCurry},
so it is left out of the natural semantics of Curry~\cite{currySemantics}.
However, we will handle apply with 
the standard eval-apply method~\cite{evalApply}.

We also make a few changes from previous presentations of the
syntax~\cite{currySemantics,TheoryToCurry}.
We include failure as the value $\bot$.
This allows us to encode failing computations 
explicitly in our execution model,
and to show how failure propagates throughout a computation.
We also have a separate construct for declaring free variables
because the standard implementation of FlatCurry
also has a separate construct.

Finally, literals, primitive operations, and residuation are
all outside the scope of this paper, although not outside the 
scope of the compiler.
This is done both for brevity and because their
implementations are not novel.

\begin{figure}
\centering
\captionsetup{justification=centering}
\begin{tabular}{lcll}
$f$ & $\To$ & $\code{f} \vx = e$          & function definition \\
$e$ & $\To$ & $x$                         & Variable \\
    & $|$   & $e_1\ ?\ e_2$               & Choice \\
    & $|$   & $\bot$                      & Failed \\
    & $|$   & $\code{f}\ \vec{e}$         & Function Application \\
    & $|$   & $\code{C}\ \vec{e}$         & Constructor Application \\
    & $|$   & \clet{\vec{x = e}}{e}       & Variable Declaration \\
    & $|$   & \cfree{\vec{x}}{e}          & Free Variable \\
    & $|$   & \ccase{e}{\vec{p \to e}}    & Case Expression \\
$p$ & $\To$ & $\code{C}\ \vec{x}$         & Constructor Pattern \\
    & $|$   & $l$                         & Literal Pattern \\
\end{tabular}
    \caption{The syntax of FlatCurry
             We use the convention of $x$ for variables, 
             \texttt{f} for function names,
             and \texttt{C} for constructor names.}
    \label{fig:flat}
\end{figure}

\begin{figure}
\centering
\captionsetup{justification=centering}
\begin{tabular}{ll}
(Nat-VarCons) & 
\atom{$\Gamma[x \mapsto t] : x \Downarrow_C \Gamma[x \mapsto t] : t$}
\DisplayProof where $t$ is a head normal form\\
 & \\
(Nat-VarExp) & 
\atom{$\Gamma[x \mapsto e] : e \Downarrow_C \Delta : v$}
\unary{$\Gamma[x \mapsto e] : x \Downarrow_C \Gamma[x \mapsto v] : v$}
\DisplayProof  where $e$ is not a head normal form\\
 & \\
(Nat-Val) &
\atom{$\Gamma : v \Downarrow_C \Gamma : v$}
\DisplayProof  Where $v$ is a head normal form\\
 & \\
(Nat-Fun) &
\atom{$\Gamma : \rho(e) \Downarrow_C \Delta : v$}
\unary{$\Gamma : f\ \vec{y} \Downarrow_C \Delta : v$}
    \DisplayProof  where $f\ \vec x = e\in P$ is a defined function 
                   and $\rho(x_n) = y_n$\\
    &  \\
(Nat-Let) &
\atom{$\Gamma[\vec{y_k\mapsto \rho(e_k)}] : e \Downarrow_C \Delta : v$}
\unary{$\Gamma : $\clet{\vec{x_k = e_k}}{e}
       $\Downarrow_C \Delta : v$}
\DisplayProof where $\rho(x_n) = y_n$\\
    & \\
(Nat-Or) &
\atom{$\Gamma : e_i \Downarrow_C \Delta : v$}
\unary{$\Gamma : e_1\ ?\ e_2 \Downarrow_C \Delta : v$}
    \DisplayProof  where $i\in \{1,2\}$\\
    & \\
(Nat-Select) &
\atom{$\Gamma : e \Downarrow_C \Delta : C_i(\vec z)$}
    \atom{$\Delta : e_i[\vec{y\mapsto z}] \Downarrow_C \Theta : v$}
\binary{$\Gamma :$ \ccase{e}{\vec{C_i(\vec y) \to e_i}} 
        $\Downarrow_C \Theta : v$}
\DisplayProof \\
    & \\
(Nat-Guess) &
\atom{$\Gamma : e 
       \Downarrow_C \Delta : x$}
\atom{$\Delta[x\mapsto C_i(\vec y)][\vec{y\mapsto y}] : e_i 
       \Downarrow_C \Theta : v$}
\binary{$\Gamma :$ \ccase{e}{\vec{C_i(\vec y) \to e_i}}
         $\Downarrow_C \Theta : v$}
\DisplayProof \\
\end{tabular}
    \caption{Natural semantics for Curry~\cite{currySemantics}\\
             Following the conventions, $\Gamma[x \mapsto v]$
             can be used to lookup or update variable $x$ in heap 
             $\Gamma$ with value $v$.}
\label{fig:nat}
\end{figure}

\section{Restricted FlatCurry}

We now turn our attention to transforming FlatCurry programs
so that they are more amenable to an implementation.
Our restricted language is very similar to \brassel's flat uniform
programs~\cite{Kics2Theory},
and is similar, although not identical to, A-Normal form~\cite{anf}.
The primary restrictions of Restricted FlatCurry
are that functions and constructors are only applied to trivial arguments;
let expressions cannot be nested; and each function definition can contain
at most one case expression.
We split the syntax into three sections: 
blocks, statements, and expressions.
A block consists of a single \texttt{case}
where each of the branches contains statements.
A statement consists of zero or more 
\texttt{let} expressions, followed by an expression.
Finally, an expression can be either a variable, literal,
choice, failure, function application, 
constructor application, or \texttt{free}.
We only allow a single case for each function,
and all declarations must occur as early as possible.
Furthermore, all applications, including function, constructor, and choice,
must be applied to variables.
The difference between the body, statements, and expressions
is only for the purposes of giving structure to Curry functions.
An example of a FlatCurry program, and a restricted FlatCurry program
can be seen in Figure \ref{fig:flatEX}.
Throughout the rest of the paper, we will refer to everything
as an expression.
This structure closely corresponds to the structure of ICurry~\cite{icurry}.

We changed the representation of free variables in this syntax
to correspond with their role in the \rice{} runtime.
A free variable is a normal form that case expressions
can narrow.
It would be perfectly sensible to replace the free variable with a generator
at this point~\cite{TheoryToCurry}, but \rice{} implements narrowing.

\begin{figure}
\centering
\captionsetup{justification=centering}
\textbf{Curry}
\begin{curry}
and False _ = False
and True False = False
and True True = True
\end{curry}
\textbf{FlatCurry}
\begin{curry}
and x y = case x of
               False -> False
               True -> case y of
                            False -> False
                            True -> True
\end{curry}
\textbf{restricted FlatCurry}
\begin{curry}
and x y = case x of
               False -> False
               True -> and1 y

and1 y = case y of
              False -> False
              True -> True
\end{curry}

\caption{A Curry function, a FlatCurry function, and a restricted FlatCurry
    function.}
\label{fig:flatEX}
\end{figure}

\begin{figure}
\centering
\captionsetup{justification=centering}
\begin{tabular}{lcll}
$f$ & $\To$ & $f\ \vec{x} = b$            & Function Definition \\
$b$ & $\To$ & \ccase{x}{\vec{p \to s}}    & Case Expression \\
    & $|$   & $s$                         & statement \\
$s$ & $\To$ & \clet{\vec{x = e}}{s}       & Variable Declaration \\
    & $|$   & $e$                         & Return Expression \\
$e$ & $\To$ & $x$                         & Variable  \\
    & $|$   & $l$                         & Literal \\
    & $|$   & $\bot$                      & Failed \\
    & $|$   & $x\ ?\ x$                   & Choice \\
    & $|$   & \code{free}                 & Free variable \\
    & $|$   & $f\ \vec{x}$                & Function Application \\
    & $|$   & $C\ \vec{x}$                & Constructor Application \\
    & $|$   & $\code{apply}\ x\ \vec{x}$  & Application \\
$p$ & $\To$ & $C\ \vec{x}$                & Constructor Pattern \\
    & $|$   & $l$                         & Literal Pattern \\
\end{tabular}
\caption{Syntax of restricted FlatCurry}
\label{ref:restrict}
\end{figure}

\section{Heap Representation}

Now that we have a syntax for Curry, we can discuss the execution model.
A Curry program consists of a set of functions as well as a single 
expression to evaluate.  
The expression is represented as a directed rooted graph 
that we will continually reduce.
The graph plays the same role as the heap in traditional
implementations of functional languages~\cite{continuationsAppel,stg}.
We refer to it as a graph to stay closer in line with the theory of Curry.

The graph nodes are given in Figure \ref{fig:graph}.
We use the notation $\code{f}(x)$ to represent a function application
to distinguish it from the FlatCurry syntax.
Specifically, this represents the node $\code f$, with a single child $x$.
If $G$ is a graph with node $n$, then $G[n]$ refers to the subgraph
rooted by node $n$, and $G[n \mapsto g]$ means 
replace node $n$ in $G$ with the graph $g$.
We also use the convention that 
if a node $n$ is referred to more than once, then it is shared.
For example, if $G[n] = True$, then $xor(n,n)$ refers to the following graph.
$$
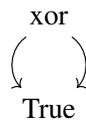
\begin{tikzcd}
    \text{\code{xor}} \ar[d, bend right=50] \ar[d, bend left=50] \\
    \text{\code{True}}
\end{tikzcd}
$$

We discuss the different nodes below.
The graph contains seven different types of nodes:
$\bot$, \code{free}, \code{?}, function, constructor,
forwarding, and partial application nodes..

The $\bot$ and \code{free} nodes are the most direct nodes.
The $\bot$ node represents a failing computation,
and only serves to propagate the failure to the root of the expression.
The \code{free} node represents a free variable.

The \code{?} node represents a choice.
We treat choices in a similar way to constructors.
We do not immediately choose a branch, but instead defer it
until the value is demanded by a case expression.
However, this is only an apparent difference.
Any expression that is reduced to a choice will immediately demand
the value of that choice.
This simplifies the execution model, and there does not seem
to be a measurable performance loss for delaying the evaluation of choice.

Function and constructor applications are always fully applied.
For partial applications, we have the \texttt{PART} node.
This node contains three things: a function or constructor to apply,
a number $k$ representing the number of arguments the function is missing,
and finally the arguments that have already been partially applied.
In the implementation, the function is represented by a closure.

Finally, $\FWD(g)$ represents a forwarding, or indirection, node.
The notation is supposed to resemble a pointer.
Forwarding nodes are necessary for a function that 
returns one of its parameters.
Consider the following example.

\begin{curry}
id x = x

main = let x = True ? False
       in xor (id x) (id x)
\end{curry}

The graph representing our expression is:
$$
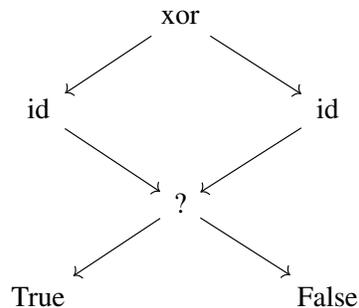
\begin{tikzcd}
    & \text{\code{xor}} \ar[dl] \ar[dr] & \\
    \text{\code{id}} \ar[dr] & & \text{\code{id}}  \ar[dl] \\
    & \text{\code{?}} \ar[dl] \ar[dr] & \\
    \text{\code{True}} & & \text{\code{False}} \\
\end{tikzcd}
$$

If we take an approach similar to 
Peyton Jones~\cite{lazyFunctionalCompilers} and 
copy the constructor after evaluation, we end up with the following graph.
$$
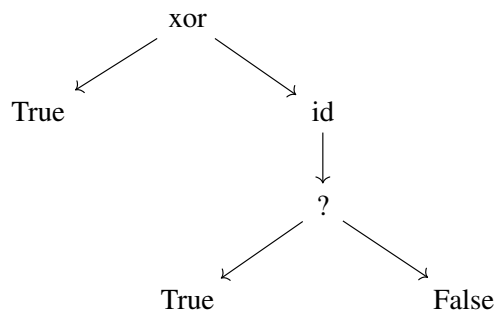
\begin{tikzcd}
    & \text{\code{xor}} \ar[dl] \ar[dr] & \\
    \text{\code{True}} & & \text{\code{id}}  \ar[d] \\
    & & \text{\code{?}} \ar[dl] \ar[dr] & \\
    & \text{\code{True}} & & \text{\code{False}} \\
\end{tikzcd}
$$

This will certainly cause problems as we try to backtrack,
because we need to replace both copies of \code{True}.
Instead, we solve this problem with the forwarding node $\FWD(g)$.
If a Curry function evaluates to a parameter of the function,
then we construct a forwarding node to maintain the structure of the graph,
and prevent any unintended copying.
Our example from before evaluates to the following.
$$
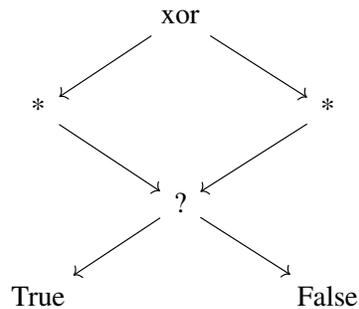
\begin{tikzcd}
    & \text{\code{xor}} \ar[dl] \ar[dr] & \\
    \text{\FWD} \ar[dr] & & \text{\FWD}  \ar[dl] \\
    & \text{\code{?}} \ar[dl] \ar[dr] & \\
    \text{\code{True}} & & \text{\code{False}} \\
\end{tikzcd}
$$

\begin{figure}
\centering
\captionsetup{justification=centering}
\begin{tabular}{lcll}
$g$ & $\To$ & $\bot$                     & Failed \\
    & $|$   & \code{free}                & Free variable \\
    & $|$   & $\code{?}(g,g)$            & Choice \\
    & $|$   & $\code{f}(\vec{g})$        & Function Application \\
    & $|$   & $\code{C}(\vec{g})$        & Constructor Application \\
    & $|$   & $\FWD(g)$                  & Forwarding Node \\
    & $|$   & $\code{PART}(\code{f},k,\vec{g})$ & Partial Application Node \\
\end{tabular}
\caption{Heap objects represented as a graph.}
\label{fig:graph}
\end{figure}

\section{The Execution Model}

To run a Curry program, we evaluate the expression \texttt{main}
to normal form (or a value).
We can accomplish this by successively evaluating 
\texttt{main} to a head normal form, 
which is a form rooted by a constructor or literal.
We then successively evaluate the children of \texttt{main} to normal form.
If \texttt{main} evaluates to a value, then we display it to the user;
if it evaluates to $\bot$, we discard the value.
In either case, we backtrack and try again.
The backtracking scheme is well understood and used 
in both \pakcs\ and \mcc~\cite{pakcs,mcc,TheoryToCurry}.

In order to implement backtracking,
we need to keep track of a backtracking stack.
We represent the backtracking stack as 
a list of frames enclosed in angle brackets.

\begin{tabular}{lll}
    $S$ & $=$ & $\langle \rangle $ \\
    $S$ & $=$ & $\langle g_{[?]},g | S\rangle$
\end{tabular}

A stack is empty, or it contains two nodes from the heap.
The left node is the current value in the heap, 
and the right node contains a value to replace it with when backtracking.
We use the notation $g_?$ to denote that node $g$ came from a choice,
and therefore backtracking should stop at this node.

We introduce four relations in Figure \ref{fig:rel}
for evaluation to normal form, head normal form,
backtracking a single step, and backtracking to a choice.
The normal form relation evaluates an expression to a value
as described above.  The graph $G$ and stack $S$ may be changed over
the course of evaluation.
The evaluation to head normal form is similar, but only evaluates $e$
to an expression rooted by a constructor.
Finally, we have two backtracking relations.
The first only undoes a single rewrite from the stack.
The second one pops rewrites off that stack until we reach a rewrite
that came from a choice.
The rules for evaluating to normal form and backtracking are given in
Figure \ref{fig:norm}.
We use the standard $\Downarrow^n$ notation to refer to the n-fold
composition of the $\Downarrow$ relation.
Most rules are standard, but the rule for choice may be surprising.
If an expression is evaluated to a choice,
then we choose the left hand side, and evaluate that to normal form.
We push the right hand side on the stack so we can backtrack 
to it later.

\begin{figure}
\centering
\captionsetup{justification=centering}
\begin{tabular}{ll}
    $G,S : e \Downarrow_N G,S : v$ & evaluation to Normal Form \\
    $G,S : e \Downarrow G,S : v$ & evaluation to Head Normal Form \\
    $G,S \Downarrow_B G,S$ & Backtracking \\
    $G,S \Downarrow_{B?} G,S$ & Backtracking to a choice\\
\end{tabular}
\caption{evaluation relations\\
         expression $e$ with graph $G$ and stack $S$ evaluates to
         value $v$ with a possibly modified $G$ and $S$.}
\label{fig:rel}
\end{figure}

\begin{figure}
\centering
\captionsetup{justification=centering}
\begin{tabular}{ll}
    (BT) &
   \atom{$G,\langle x,y | S \rangle 
          \Downarrow_B G[x\mapsto y], S$}
   \DisplayProof \\
    & \\
    (BT-Choice) &
   \atom{$G,\langle \vec{x,y}\ |\ l_?,r\ |\ S \rangle 
          \Downarrow_{B?} G[\vec{x\mapsto y}][l\mapsto r], 
          \langle r, ?(l,r)|S\rangle$}
    \DisplayProof \\
    & \\
    (Norm-Bot) &
    \atom{$G,S:e \Downarrow G_1,S_1:\bot$}
    \unary{$G,S:e \Downarrow_N G_1,S_1:\bot$}
    \DisplayProof \\
    & \\
    (Norm-Lit) &
    \atom{$G,S:e \Downarrow G_1,S_1:l$}
    \unary{$G,S:e \Downarrow_N G_1,S_1:l$}
    \DisplayProof \\
    & \\
    (Norm-Free) &
    \atom{$G,S:e \Downarrow G_1,S_1:\tfree$}
    \unary{$G,S:e \Downarrow_N G_1,S_1:\tfree$}
    \DisplayProof \\
    & \\
    (Norm-Con) &
    \atom{$G,S:e \Downarrow G_0,S_0:C(\vec{e})$}
    \atom{$\vec{G_i,S_i:e_i \Downarrow_N G_{i+1},S_{i+1}:v_i}$}
    \binary{$G,S:e \Downarrow_N G_n,S_n:C(\vec{v})$}
    \DisplayProof \\
    & \\
    (Norm-Choice) &
    \atom{$G,S:e \Downarrow G_1,S_1:?(x,y)$}
    \atom{$G_1[r\mapsto \FWD(x)],\langle r_?, \FWD(y) | S_1 \rangle:x 
          \Downarrow_N G_2,S_2:v$}
    \binary{$G,S:e \Downarrow_N G_2,S_2:\FWD(v)$}
    \DisplayProof \\
\end{tabular}
\caption{backtracking and normalization algorithm.\\
         in (Norm-Choice) $r$ is the root of the expression $e$}
\label{fig:norm}
\end{figure}

While backtracking and evaluation to normal form are typical,
evaluation to head normal form requires more explanation.
We split the rules up into three parts: basic rules in Figure \ref{fig:basic},
rules for a \texttt{case} in Figure \ref{fig:case}, and rules for \texttt{apply}
in Figure \ref{fig:apply}.

The basic rules correspond closely to the
original FlatCurry semantics with the addition of the stack.
The rules are given in Figure \ref{fig:basic}.
The rules for (Bot), (Lit), (Free), and (Con) are already in head normal form,
so the evaluation is complete.
We also treat $?$ and \FWD\ as head normal forms.
This is still consistent with the previous semantics because
they will both be evaluated by case expressions.
The rule for (Let) simply adds each defined variable to the graph.
Because all functions are only applied to variables,
we treat expression variables the same as graph nodes.
The case for (Fun) is very similar to the previous semantics.
We replace the function call with the expression graph from the function's
definition and continue evaluation.

Finally, the (Var) case is trivial.
This seems surprising because the (Nat-VarExp) case was more complicated
in the previous semantics.
However, the only way an expression could evaluate to a variable that
was not the scrutinee of a case expression is if a function
returned one of its parameters.
In that case, we need to create a forwarding 
node for the reasons described above.

\begin{figure}
\centering
\captionsetup{justification=centering}
\begin{tabular}{lll}
(Bot) &
\atom{$G,S : \bot \Downarrow G, S : \bot$}
\DisplayProof \\
    & & \\
(Lit) &
\atom{$G,S : l \Downarrow G, S : l$}
\DisplayProof \\
    & & \\
(Free) &
\atom{$G,S : \tfree \Downarrow G, S : \tfree$}
\DisplayProof \\
    & & \\
(Con) &
\atom{$G,S : C\ \vec{e} \Downarrow G, S : C(\vec{e})$}
\DisplayProof \\
    & & \\
(Choice) &
\atom{$G,S : e_1\ ?\ e_2 \Downarrow G, S :\ ?(e_1,e_2)$}
\DisplayProof \\
    & & \\
(Fun) &
\atom{$f\ \vec{x} = e$}
\atom{$G, S : e[\vec{x \mapsto y}] \Downarrow G_1, S_1 : v$}
\binary{$G,S : f\ \vec{y} \Downarrow G_1, S_1 : v$}
\DisplayProof \\
    & & \\
(Let) &
\atom{$G[\vec{x\mapsto e}], S : e_1 \Downarrow G_1, S_1 : v$}
\unary{$G, S :$ \clet{\vec{x = e}}{e_1}
       $\Downarrow G_1, S_1 : v$}
\DisplayProof \\
    & & \\
(Var) &
\atom{$G, S : x \Downarrow G, S : \FWD(x)$}
\DisplayProof \\
    & & \\
\end{tabular}
\caption{Evaluation of expressions without \texttt{case}\\
         We assume all variables from function definitions are fresh.}
\label{fig:basic}
\end{figure}

More substantial changes start to appear in the case rules.
These rules correspond to the while/switch loop in the generated
C code for the \rice{} compiler~\cite[Chapter 4]{rice}.
Cases are only applied to variables,
so case expressions inspect the variable and evaluate it if necessary.
There is one case for each type of heap object we might scrutinize,
except for partial applications.  
Typing rules prevent partial applications from 
appearing as the scrutinee of a case.

The (Case-Bot) rule is the simplest rule; it only propagates the $\bot$ up.
The (Case-Fwd) rule unwraps its argument and tries again.
The (Case-Fun) rule evaluates a function, and updates the variable when it
returns.
This rule pushes $x$ with its old value $f(\vec y)$ onto the stack,
because that rewrite might need to be undone for backtracking.
(Case-Choice) will always choose the left-hand side,
and push the right-hand side as a non-deterministic rewrite 
$\langle x_?, \FWD(z)|S\rangle$ onto the backtracking stack.
We then update $x$ to be a forwarding node to the left-hand side $y$.
(Case-Lit) and (Case-Con) can select a branch for the case.
(Case-Con) has to replace the parameters of the constructor with 
the arguments.
Finally, (Case-LitFree) and (Case-ConFree) handle narrowing steps.
The free variable is instantiated to the pattern of the first branch.
If the branch is a constructor, then the children are filled with 
free variables.
We use the notation $e[\vec{y_i \mapsto \tfree_i}]$ to denote replacing each
free variable $y_i$ in expression $e$ with the corresponding
logic variable that was created in $C(\vec{\tfree})$.
The rest of the patterns are all pushed
onto the backtracking stack as rewrites for the free variable.

\begin{figure}
\centering
\captionsetup{justification=centering}
\begin{tabular}{lll}
(Case-Bot) &
$G[x \mapsto \bot],S :$ \ccase{x}{\vec{p \to e}} 
$\Downarrow G, S : \bot$ \\
& & \\
(Case-Fwd) &
\atom{$G,S :$ \ccase{y}{\vec{p \to e}} $\Downarrow G_1, S_1 : v$}
\unary{$G[x\mapsto \FWD(y)],S :$ \ccase{x}{\vec{p \to e}} 
    $\Downarrow G_1, S_1 : v$}
\DisplayProof \\
    & & \\
(Case-Fun) &
\atom{$G,S : f\ \vec{y} \Downarrow G_1,S_1 : v_x $}
\noLine
\unary{$G_1[x\mapsto v_x],
       \langle x,f(\vec{y}) | S_1 \rangle :$
       \ccase{v_x}{\vec{p \to e}}
      $\Downarrow G_2,S_2 : v $}
\unary{$G[x \mapsto f(\vec{y})],S :$ \ccase{x}{\vec{p \to e}}
         $\Downarrow G_2, S_2 : v$}
\DisplayProof \\
    & & \\
(Case-Choice) &
\atom{$G[x\mapsto \FWD(y)],
       \langle x_?, \FWD(z)|S\rangle :$
       \ccase{y}{\vec{p \to e}}
       $\Downarrow G_1,S_1 : v $}
\unary{$G[x \mapsto ?(y,z)],S :$ \ccase{x}{\vec{p \to e}}
         $\Downarrow G_1, S_1 : v$}
\DisplayProof \\
    & & \\
(Case-Lit) &
\atom{$G,S : e_i \Downarrow G_1,S_1 : v$}
\unary{$G[x\mapsto l_i],S :$ \ccase{x}{\vec{l \to e}}
       $\Downarrow G_1, S_1 : v$}
\DisplayProof \\
    & & \\
(Case-LitFree) &
\atom{$G[x\mapsto l_1],S_L : e_1 
       \Downarrow G_1,S_1 : v$}
\unary{$G[x\mapsto \tfree],S :$ \ccase{x}{\vec{l \to e}}
         $\Downarrow G_1, S_1 : v$}
\DisplayProof \\
    & & \\
(Case-Con) & 
\atom{$G,S : e_i[\vec{y \mapsto z}]
      \Downarrow G_1,S_1 : v$}
\unary{$G[x\mapsto C_i(\vec{z})],S :$ \ccase{x}{\vec{C\ \vec{y} \to e}}
        $\Downarrow G_1, S_1  : v$}
\DisplayProof \\
    & & \\
(Case-ConFree) & 
\atom{$G[x\mapsto C_1(\vec{\tfree})], S_C 
       : e_1[\vec{y_i \mapsto \tfree_i}]
      \Downarrow G_1,S_1 : v$}
\unary{$G[x\mapsto \tfree],S :$ \ccase{x}{\vec{C\ \vec{y} \to e}}
        $\Downarrow G_1, S_1 : v$}
\DisplayProof \\
\end{tabular}
\caption{rules for Case expressions.\\
         In Case-LitFree $S_L = \langle x,l_2| \ldots |x,l_n|
                                        x,\tfree|S_1\rangle$ \\
         In Case-ConFree $S_C = \langle x,C_2(\vec{\tfree})| \ldots|
                                        x,C_n(\vec{\tfree})| 
                                        x,\tfree | S_1 \rangle$
               }
\label{fig:case}
\end{figure}

The final rules are the rules for the \texttt{apply} function.
We handle apply with the eval-apply method~\cite{evalApply}.
In fact, 
the logic features have no bearing on partial application.
If we are applying a choice node, then we select the leftmost node
and try again.
If we are applying a free variable, then we fail.

The remaining possibilities of partial application are split into three cases.
If the partial application is under applied, then we simply add
the new arguments to the application and move on.
If the application is correctly applied, then we evaluate the function
with the arguments.
Finally, if the application is over applied, then it must evaluate to a 
\texttt{PART}; we take the first few arguments, evaluate to the \texttt{PART},
and supply the final arguments.

\begin{figure}
\centering
\captionsetup{justification=centering}
\begin{tabular}{ll}
(Apply-Free) & 
\atom{$G[x\mapsto \tfree],S : \code{apply}\ x\ \vec{e} 
    \Downarrow G, S : \bot$}
\DisplayProof \\
    & \\
(Apply-Choice) & 
\atom{$G[x\mapsto \FWD(y)],\langle x_?, \FWD(z)|S\rangle : 
      \code{apply}\ y\ \vec{e} 
    \Downarrow G_1, S_1 : v$}
\unary{$G[x\mapsto ?(y,z)],S : \code{apply}\ x\ \vec{e} 
    \Downarrow G_1, S_1 : v$}
\DisplayProof \\
    & \\
(Apply-Under) & 
\atom{$|\vec{e}| = n < k$}
\unary{$G_x,S : \code{apply}\ x\ \vec{e} 
    \Downarrow G, S : \code{PART}(f, k-n, \vec{ye})$}
\DisplayProof \\
    & \\

(Apply-Full) & 
    \atom{$G,S : f\ \vec{y}\ \vec{e_k} \Downarrow G_1,S_1:v$}
    \unary{$G_x,S : \code{apply}\ x\ \vec{e_k}
    \Downarrow G_1, S_1 : v$}
\DisplayProof \\
    & \\
(Apply-Over) & 
\atom{$G,S : f\ \vec{y}\ \vec{e_k} \Downarrow G_1,S_1:v_1$}
    \atom{ $G_1,S_1 : \code{apply}\ v_1\ \vec{e_{k+1}}
    \Downarrow G_2,S_2:v_2$}
    \binary{$G_x,S : \code{apply}\ x\ \vec{e_k} \vec{e_{k+1}}
    \Downarrow G_2, S_2 : v_2$}
\DisplayProof \\
\end{tabular}
\caption{apply rules. \\
         In all 3 rules $G_x = G[x \mapsto \code{PART}(f,k,\vec{y})]$ \\
         In (apply-over) $\vec{e_k} = e_1\ldots e_k$ \\
         In (apply-over) $\vec{e_{k+1}} = e_{k+1}\ldots e_n$
         }
\label{fig:apply}
\end{figure}

\section{Correspondence to Generated Code}

The purpose of this semantics is to model the execution of programs
compiled with \rice{}. In order to justify that the semantics really
does correspond with the compiled code, we give a small example of 
compiled \rice{} code for the \code{not} function defined below
\begin{curry}
not x = case x of
             True -> False
             False -> True
\end{curry}

In the \rice{} compiler Curry expressions are compiled C code.
The graph nodes from the semantics are represented as \code{Node} objects,
which correspond to closures in a traditional functional language.
Its definition is given in Figure \ref{fig:node}.
Each \code{Node} has 3 fields.  The missing field is used for partial
application, the symbol field contains information about the \code{Node}
such as it's name and arity and tag describing what node it is,
as well as a pointer to code to reduce the
node to head normal form.
Finally, each node has an array of 4 children.
If a node has arity greater than 4, the final slot is a pointer
to an array containing the rest of the children.
\begin{figure}
\begin{ccode}
typedef struct Node {
    const unsigned char tag;
    int missing;
    Symbol* symbol;
    field children[4];
} Node;
\end{ccode}
\caption{Definition for a \code{Node} object}
\label{fig:node}
\end{figure}

The code for reducing the expression \code{not e} for some expression \code e
is given in Figure \ref{fig:notHNF}.
We continue to loop until we reduce to a head normal form.
The branches in the case corresponding to the different Case rules
in Figure \ref{fig:case}.
The rules for \code{FAIL}, \code{True}, and \code{False} just set the symbol,
and remove the child and return.
The forward tag sets the scrutinee to be its first child and retries.
The choice tag makes sets the scrutinee to be
one of its children and pushes in on the stack.
The details are elided here.
Finally, The function rule reduces the scrutinee and 
pushes it on the backtracking stack which we call \code{bt_stack}.

\begin{figure}
\begin{ccode}
void not_hnf(field root) {
  Node* scrutinee = root->children[0];
  while(true) {
    switch(scrutinee->tag) {
      case FAIL_TAG:
        root->scrutinee = FAIL_symbol;
        root->children[0] = NULL;
        return;
      case FORWARD_TAG:
        scrutinee = scrutinee->children[0];
        break;
      case CHOICE_TAG:
        choose(scrutinee);
        break;
      case FUNCTION_TAG:
        scrutinee->symbol->hnf(srutinee);
        push(bt_stack, scrutinee, false);
        break;
      case True_TAG:
        root->symbol = False_Symbol;
        root->children[0] = NULL;
        return;
      case False_TAG:
        root->symbol = True_Symbol;
        root->children[0] = NULL;
        return;
    }
  }
}
\end{ccode}
    \caption{Code for reducing a $not$ node to head normal form.}
    \label{fig:notHNF}
\end{figure}

\section{Correctness}

With the semantics now established, we need to show that they actually
implement Curry.
We do this via comparison to the original semantics~\cite{currySemantics}.
However, the original semantics were non-deterministic,
and we cannot hope to match them.
Because we are not using a fair evaluation strategy,
there will be answers that we may not produce in a finite amount of time.

Nevertheless, we can still show that we agree with the original semantics
in restricted cases.  Specifically, we show that for terminating programs,
we produce the same answers as the original semantics,
which is not surprising.
If we remove the stack from our semantics, then we match the original
fairly closely.

We start by proving that the backtracking 
operation does backtrack as we claim.
Specifically, we show that if we have a graph $G$ and expression $e$
that we evaluate to the new graph $G'$ and $v$,
then there is some number of backtracking steps 
that will restore the original graph.

\begin{theorem}
For any expression graph $G$, stack $S$, and expression $e$
if $G,S:e \Downarrow_N G',S':v$, 
then there exists some $n$ where 
$G',S' \Downarrow_B^n G\cup G_u, S$ where $G_u$
is a set of unreachable bindings created after $e$ was evaluated.
\end{theorem}

\begin{proof}
The proof is by structural induction on the derivation of $e$.
The only rules that alter $G$ are 
(Let), (Case-Choice), (Case-Fun). (Case-LitFree), (Case-ConFree),
(Norm-Choice), and (Apply-Choice).
All other cases are trivial because they do not modify the stack
or the graph.
The (Let) rule can only add new bindings to the graph,
so these new bindings may go in $G_u$.

(Case-Fun):
If $G[x \mapsto f(\vec y)]$, then there are two evaluations
that take place: $G,S:f(\vec y) \Downarrow G_1,S_1:v_x$
and $G[x\mapsto v_x], \langle x,f(\vec{y}) | S_1 \rangle :$
     \ccase{v_x}{\vec{p \to e}} $\Downarrow G_2,S_2 : v $.
By our inductive hypotheses 
$G_2,S_2 \Downarrow_B^n G_1[x \mapsto v_x]\cup G_{u2},
\langle x, f(\vec y) | S_1 \rangle$,
so $G_2,S_2 \Downarrow_B^{n+1} G_1[x \mapsto f(\vec{y})]\cup G_{u2},S_1$.
Again, by our induction hypothesis
$G_1,S_1 \Downarrow_B^m G\cup G_{u1},S$.
Therefore, $G_2,S_2 \Downarrow_B^{m+n+1} G\cup(G_{u1}\cup G_{u2}), S$.

The cases for (Case-Choice), (Case-LitFree), (Case-ConFree),
(Norm-Choice), and (Apply-Choice)
are all similar, but there is one slight alteration.
We use (Case-Choice) as an example.
If $G[x \mapsto ?(y,z)]$,
then there is only one derivation.
$G[x\mapsto \FWD(y)],
   \langle x_?, \FWD(z)|S\rangle :$
   \ccase{y}{\vec{p \to e}}
   $\Downarrow G_1,S_1 : v $
By the induction hypothesis 
$G_1,S_1 \Downarrow_B^{n} G[x\mapsto \FWD(y)]\cup G_n,
          \langle x_?, \FWD(z)|S\rangle$.
Our next backtracking step will replace $y$ with $z$
and add something new to the stack.
$G_1,S_1 \Downarrow_B^{n+1} G[x\mapsto \FWD(z)]\cup G_n,
          \langle x, ?(y,z)|S\rangle$.
This is fine because the next backtracking 
step will restore the stack.
$G_1,S_1 \Downarrow_B^{n+2} G\cup G_n, S$.
This completes the proof.
\end{proof}

Next, we show that for any graph $G$, stack $S$, and expression $e$,
if we evaluate using our semantics, then we produce the same value
as the natural semantics~\cite{currySemantics} assuming all choices
choose the left hand side. We use $\Downarrow_C$ as the evaluation
relation from the natural semantics. We recall the rules
for the natural semantics in Figure \ref{fig:nat}.
Because there is only a relation for head normal forms, and not normal forms,
we restrict ourselves to evaluations that terminate in a constructor or
literal.

\begin{theorem}
    If $G,S:e\Downarrow G',S':v$, Where $v$ is a constructor or a literal,
    then there is a heap $\Gamma$ that
    corresponds to $G$, and a heap $\Gamma'$ corresponding to $G'$
    such that $\Gamma : e \Downarrow_C \Gamma' : v'$
    where $v'$ is the same as $v$ with the forward nodes contracted.
\end{theorem}

\begin{proof}
We prove this by constructing a transformation on derivations in our semantics
to a derivation in the natural semantics.
Because the natural semantics is not formulated for higher order expressions,
we will assume all expressions are first order and all applications
are fully applied.

We create a mapping $\Leftrightarrow$ which maps evaluation rules from
our semantics to the natural semantics.
The full mapping can be found in Figure \ref{fig:map}.
By our assumption, (Bot) or (Case-Bot) can never appear
in the evaluation.  If they did, then the $\bot$ would propagate
to the root of the expression.
The cases for (Lit), (Con), (Free), (Fun), and (Let) are straightforward.
We elide the stack in all of these mappings
because it is not relevant to the proof.

\begin{figure}
\centering
\captionsetup{justification=centering}
\begin{tabular}{ll}
(Lit) &
\atom{$G : l \Downarrow G : l$}
\DisplayProof \\
$\Leftrightarrow$     & \\
(Nat-Val) &
\atom{$\Gamma : l \Downarrow \Gamma : l$}
\DisplayProof \\ \hline
    &  \\
(Con) &
\atom{$G : C\ \vec{e} \Downarrow G : C(\vec{e})$}
\DisplayProof \\
$\Leftrightarrow$    &  \\
(Nat-Val) &
\atom{$\Gamma : C\ \vec e \Downarrow \Gamma : C\ \vec e$}
\DisplayProof \\ \hline
    &  \\
(Fun) &
\atom{$f\ \vec{x} = e$}
\atom{$G : e[\vec{x \mapsto y}] \Downarrow G_1 : v$}
\binary{$G : f\ \vec{y} \Downarrow G_1 : v$}
\DisplayProof \\
$\Leftrightarrow$    & \\
(Nat-Fun) &
\atom{$\Gamma : \rho(e) \Downarrow \Delta : v$}
\unary{$\Gamma : f\ \vec{y} \Downarrow \Delta : v$}
    \DisplayProof  where $f\ \vec x\in P$ and $\rho(y_n) = x_n$\\ \hline
    &  \\
(Let) &
\atom{$G[\vec{x\mapsto e}] : e_1 \Downarrow G_1 : v$}
\unary{$G :$ \clet{\vec{x = e}}{e_1}
       $\Downarrow G_1 : v$}
\DisplayProof \\
$\Leftrightarrow$    & \\
(Nat-Let) &
\atom{$\Gamma[\vec{y_k\mapsto \rho(e_k)}] : e \Downarrow \Delta : v$}
\unary{$\Gamma : $\clet{\vec{x_k = e_k}}{e}
       $\Downarrow \Delta : v$}
    \DisplayProof  where $\rho(x_k) = y_k$ \\ \hline
(Case-Fun-Lit) &
\atom{$G : f\ \vec{y} \Downarrow G_1 : l_i $}
\atom{$G_1 : e_i \Downarrow G_2 : v$}
\unary{$G_1[x\mapsto l_i] :$
       \ccase{v_x}{\vec{l \to e}}
      $\Downarrow G_2 : v $}
\binary{$G[x \mapsto f(\vec{y})] :$ \ccase{x}{\vec{l \to e}}
         $\Downarrow G_2 : v$}
\DisplayProof \\
$\Leftrightarrow$    & \\
(Nat-Select) &
\atom{$\Gamma : \rho(e) \Downarrow \Delta : l_i$}
\unary{$\Gamma : f\ \vec y \Downarrow \Delta : l_i$}
\atom{$\Delta : e_i \Downarrow \Phi : v$}
\binary{$\Gamma[x\mapsto f\ \vec y] :$ \ccase{x}{\vec{l_i \to e_i}} $\Downarrow \Phi : v$}
\DisplayProof where $f\ \vec x\in P$ and $\rho(y_n) = x_n$\\ \hline
 & \\
(Case-Choice) &
\atom{$G : f\ \vec{y} \Downarrow G_1 : ?(y,z) $}
\atom{$G_1: e \Downarrow G_2 : p_i $}
\atom{$G_2: e_i \Downarrow G_3 : v $}
\binary{$G_1[y\mapsto e]:$
       \ccase{y}{\vec{p \to e}}
       $\Downarrow G_3 : v $}
\unary{$G_1[x\mapsto \FWD(y)]:$
       \ccase{x}{\vec{p \to e}}
       $\Downarrow G_3 : v $}
\unary{$G_1[x \mapsto ?(y,z)] :$ \ccase{x}{\vec{p \to e}}
         $\Downarrow G_3 : v$}
\binary{$G_1[x \mapsto f(\vec{y})] :$ \ccase{x}{\vec{l \to e}}
         $\Downarrow G_3 : v$}
\DisplayProof \\
$\Leftrightarrow$ & \\
(Nat-Or) &
\atom{$\Gamma : y \Downarrow \Delta : p_i$}
\unary{$\Gamma : y\ or\ z \Downarrow \Delta : p_i$}
\noLine
\unary{$\vdots$}
\noLine
\unary{$\Gamma : \rho(e) \Downarrow \Delta : p_i$}
\unary{$\Gamma : f\ \vec y \Downarrow \Delta : p_i$}
\atom{$\Delta : e_i \Downarrow \Phi : v$}
\binary{$\Gamma[x\mapsto f\ \vec y] :$ \ccase{x}{\vec{l_i \to e_i}} $\Downarrow \Phi : v$}
\DisplayProof \\ \hline
& \\
(Case-LitFree) &
\atom{$G[x\mapsto l_1] : e_1 
       \Downarrow G_1 : v$}
\unary{$G[x\mapsto \tfree] :$ \ccase{x}{\vec{l \to e}}
         $\Downarrow G_1 : v$}
\DisplayProof \\
    $\Leftrightarrow$ & \\
(Nat-Guess) &
\atom{$\Gamma : e 
       \Downarrow \Delta : x$}
\atom{$\Delta[x\mapsto l_1] : e_1 
       \Downarrow \Theta : v$}
\binary{$\Gamma[x\mapsto e] :$ \ccase{x}{\vec{l \to e}}
         $\Downarrow \Theta : v$}
\DisplayProof \\
\end{tabular}
\caption{The mapping $\Leftrightarrow$}
\label{fig:map}
\end{figure}

The only two non-case rules that do not directly correspond are
(Choice) and (Var).
By our assumption that $v$ is a constructor or a literal,
we know that these rules must appear in the context of 
a case expression.
Because the scrutinee of all case statements is a variable,
all of our case rules will correspond to multiple rules
in the natural semantics.
Specifically, every scrutinee that is not in head normal form 
will have a (Case-Fun) rule to evaluate it.
We can assume that there may be (Case-Fwd) rules before any of the case
rules are applied.  This does not affect the result because the forwarding
nodes will disappear after contraction.
We show the case for a function application that evaluates to a literal,
but the case for a function evaluating to a constructor is identical.

In the case of (Choice), the correspondence is not immediately clear.
The evaluation seems very different because we treat choice as
a head normal form and the natural semantics does not.
However, because all choices must be evaluated in a case,
the next step in the evaluation is to select a branch for the choice.

Finally, we will consider the narrowing step.
This is similar to choice in that free variables are normal forms,
but it is an easier correspondence because free variables are normal
forms in the natural semantics as well.
In the natural semantics, if an expression $e$ evaluates to a variable $x$,
then it must be the case that $\Gamma[x \mapsto x]$ 
and $x$ is a free variable.
We show the case for case expressions with literal branches,
but the constructor
case is identical.
Because this covers every rule, $\Leftrightarrow$ is a correspondence
between our semantics and the natural semantics.
\end{proof}

These two theorems justify the correctness of our execution model.
If we have a terminating expression $e$, and $e$ evaluates to a value
$v$ in the natural semantics, then it will eventually evaluate to $v$
in our semantics.

\section{Related Work and Conclusion}
This work was built on the work of Hanus et al.~\cite{currySemantics},
and \brassel~\cite{Kics2Theory}.
Our execution model follows the execution model of \pakcs~\cite{pakcs},
with improvements for performance.
There are a number of different semantics for Curry including
CRWL~\cite{crwl}, and rewriting~\cite{TheoryToCurry}.
We elected to go with the natural semantics because it closely
resembles the implementation of the \rice{} compiler.
Other execution models have been described for \mcc~\cite{mcc}
and \kics~\cite{kics2}.
We cannot directly use the work on \kics\ because it uses pull-tabbing
rather than backtracking.
The execution model in \mcc\ was different enough that
we did not feel it was useful to build on it.

Another alternative would be the semantics 
given for the Verse Calculus~\cite{verse}.
While we think it would be an interesting idea to
compile Curry to VC and see how the performance compares,
we still have many questions about implementation details.

In future work, we would like to show the correctness of
some of the optimizations to the execution model found in the \rice{} 
compiler~\cite{rice}.
These include fast backtracking~\cite{ricePaper} and 
case shortcutting~\cite{rice}.
We would also like to show a correspondence with the denotational semantics
given by Mehner et al.~\cite{ParametricityCurry} to use 
the free theorems to justify
some tricky compiler transformations.

We have presented the execution model for \rice{}.
We extended the natural semantics by making 
it deterministic and adding a stack.
We justified our execution model by showing that 
expressions evaluate to the same values as the natural semantics.
We believe that this execution model is simple enough to be understandable,
but detailed enough to be useful.

\bibliographystyle{eptcs}
\bibliography{lsfa24}
\end{document}